%% file: griduso.tex
\documentclass[a4paper,USenglish,numberwithinsect]{lipics-v2021}

\usepackage{microtype} 
\input{generalSettings}
\graphicspath{{./figs/}}

\usepackage{enumitem}
\newcommand{\solution}{\f{s}\xspace}

\title{Unique Sink Orientations of Grids is in Unique End of Potential Line}
\titlerunning{\GridUSO is in \UEOPL}

\author{Michaela Borzechowski}{Institut f\"ur Informatik, Freie Universität 
Berlin, Germany}{michaela.borzechowski@fu-berlin.de}{}
{Supported by DFG within the Research Training Group GRK~2434 \emph{Facets of Complexity}.}
\author{Wolfgang Mulzer}{Institut f\"ur Informatik, Freie Universität 
Berlin, Germany}{mulzer@inf.fu-berlin.de}{https://orcid.org/0000-0002-1948-5840}
{Supported in part by ERC StG 757609.}
\authorrunning{M. Borzechowski and W. Mulzer}

\hideLIPIcs
\nolinenumbers
\keywords{Unique sink orientation, UEOPL}
\ccsdesc[500]{Theory of computation~Computational geometry}

\begin{document}

\maketitle

\begin{abstract}
The complexity classes \emph{Unique End of Potential Line} (\UEOPL) and its promise version \PUEOPL were introduced in 2018 by Fearnly et al.~\cite{UEOPL2020}. \PUEOPL captures search problems where the instances are promised to have a unique solution. \UEOPL captures total search versions of these promise problems. The promise problems can be made total by defining \emph{violations} that are returned as a short certificate of an unfulfilled promise.

\GridUSO is the problem of finding the sink in a grid with a unique sink orientation. It was introduced by Gärtner et al.~\cite{GridUSO}. We describe a promise preserving reduction from \GridUSO to \UniqueForwardEOPL, a \UEOPL-complete problem. Thus, we show that \GridUSO is in \UEOPL and its promise version is in \PUEOPL.
\end{abstract}

\section{Introduction}

Many tasks in computer science are naturally formulated as \emph{search problems}, where the goal is to \emph{find} a ``solution'' for a given instance. \emph{Promise problems}, where it is guaranteed that we see only instances that have a certain property, are also an intuitive approach to formulate certain tasks. Nonetheless, standard complexity theory works with decision problems, and it may happen that the computational complexity of the decision problem and the search problem are not equivalent.
In particular, this is the case for problems for which it is guaranteed that a solution always exists.
The complexity of such \emph{total} search problems has been studied since at least 1991, when Megiddo and Papadimitriou defined the class \emph{Total Function NP} (\TFNP)~\cite{Weiss2}.
Here, we consider a subclass of \TFNP, namely the complexity class \emph{Unique End of Potential Line} (\UEOPL). It contains some interesting problems from computational geometry for which no polynomial time algorithm is known but which are unlikely to be \NP-hard, for example \myProblem{$\alpha$-Ham-Sandwich} \cite{HamSandwich} and \myProblem{Arrival} \cite{ARRIVAL3}. Currently \UEOPL contains one complete problem: \myProblem{One-Permutation-Discrete-Contraction}.
The problem \GridUSO is an abstraction of the simplex algorithm executed on a \emph{general P-Matrix linear complementarity problem}.
Since \UEOPL was introduced only recently, in 2018, by Fearnly et al.~\cite{UEOPL2020}, the knowledge about this class is still very limited. 
We show that the problem \GridUSO lies in \UEOPL, making progress towards elucidating the nature of this class
and the problems in it. In particular, with more problems that are known to lie in \UEOPL, it becomes more likely that additional complete problems are found.

\section{Unique End of Potential Line}

A \UniqueForwardEOPL instance is defined by two circuits $\Successor: \{0, 1\}^d \rightarrow \{0, 1\}^\dimension$ and $\cost: \{0, 1\}^d \rightarrow \{0, 1\}^m$.
A circuit is a compact polynomial sized representation of information which in a map would be exponentially large (for a more formal definition, see \cite[Definition 6.1]{ComputationalComplexityModern}).
The circuits represent a directed graph \Graph. The vertices of \Graph are bit strings $\nodeA \in \{0, 1\}^\dimension$ with $\Successor(\nodeA) \neq \nodeA$, and there is a directed edge from a node \nodeA to a node \nodeB if and only if $\Successor(\nodeA) = \nodeB$ and $c(\nodeB) > c(\nodeA)$ (where the result of \cost is interpreted as an $m$-bit positive integer). 
Thus, each node in \Graph has out-degree of at most $1$, and the circuit \Successor computes the candidate \emph{successor} of a node. The circuit $\cost$ assigns a positive \emph{cost} (also called \emph{potential}) from $\{0, \dots, 2^m - 1\}$ to every node, and all edges go in the direction of strictly increasing cost.
Furthermore, we define that the bit string $0^\dimension$ is a node of \Graph, and that $\cost(0^\dimension) = 0$.
These properties ensure that \Graph is a collection of directed paths, which we call \emph{lines}, and that $0^\dimension$ is the start vertex of a line.
Our computational task is as follows: if the nodes of \Graph form a single line (that necessarily starts in $0^\dimension$), then we should find the unique end node of this line --- the \emph{sink}. 
Otherwise, we should find a \emph{violation certificate} that shows that \Graph does not consist of a single line. 
\UniqueForwardEOPL is a total search problem. There always exists a valid sink or a violation. 
Note that there might exist a valid sink and a violation simultaneously.
The promise version of \UniqueForwardEOPL is: under the promise that no violations exist for the given instance, find the unique end of the line.
The formal definition of the total search problem \UniqueForwardEOPL is as follows:

\begin{definition} (\cite[Definition~10]{UEOPL2020})
Let \f{\dimension, m \in \Naturals^+} with \f{m\geq \dimension.}
Given Boolean circuits 
\f{\Successor : \{0, 1\}^\dimension \rightarrow \{0, 1\}^\dimension} and 
\f{\cost: \{0, 1\}^\dimension \rightarrow \{0, 1, \dots, 2^m -1 \}} which must have the property that 
\f{\Successor(0^\dimension)\neq 0^\dimension}  and
\f{\cost(0^\dimension) = 0}, find one of the following:

\begin{itemize}[leftmargin=1.5cm]
\item[\UFEOPLSolTypeEndOfLine] A bit string \f{\nodeA \in \{0, 1\}^\dimension} 
with 
\f{\Successor(\nodeA) \neq \nodeA} and either
\f{\Successor(\Successor(\nodeA)) = \Successor(\nodeA)}
or
\f{\cost(\Successor(\nodeA)) \leq \cost(\nodeA)}. 
Then, $\Successor(\nodeA)$ is not a valid node and \nodeA is a sink node in \Graph and thus a valid solution.
\item[\UFEOPLVioTypeNotALine] Two bit strings \f{\nodeA, \nodeB \in \{0, 1\}^\dimension} with \f{\nodeA \neq \nodeB}, 
\f{\Successor(\nodeA)\neq \nodeA} , \f{\Successor(\nodeB) \neq \nodeB} and either 
\emph{(a)} \f{\cost(\nodeA)=\cost(\nodeB}) or 
\emph{(b)} \f{\cost(\nodeA) < \cost(\nodeB) < \cost(\Successor(\nodeA))}.
Then, \nodeA and \nodeB are two different nodes that violate the promise of the strictly increasing potential.

\item[\UFEOPLVioBreakInLine] Two nodes \f{\nodeA,\nodeB \in \{0, 1\}^\dimension} such that \nodeA is a solution of type \UFEOPLSolTypeEndOfLine, \f{\nodeA \neq \nodeB}, \f{\Successor(\nodeB)\neq \nodeB} and \f{\cost(\nodeA) < \cost(\nodeB)}.
This encodes a break in the line. The node \nodeA is the end of one line, but there exists a different line with a node \nodeB that has higher cost than \nodeA.
\end{itemize}
Two examples of an instances with no violations are shown in Figure \ref{NoVios} and an instance with all types of violations can be seen in Figure \ref{Vios}.
\end{definition}

\begin{figure}
	
	Candidate solution space with \f{\dimension=2}.
	\centering
	\includegraphics{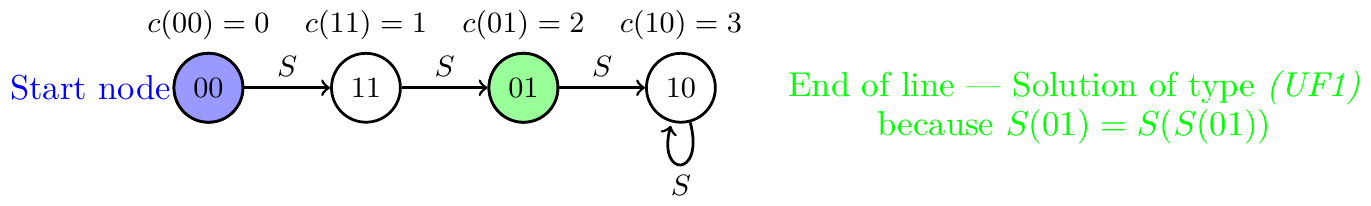}
	\includegraphics{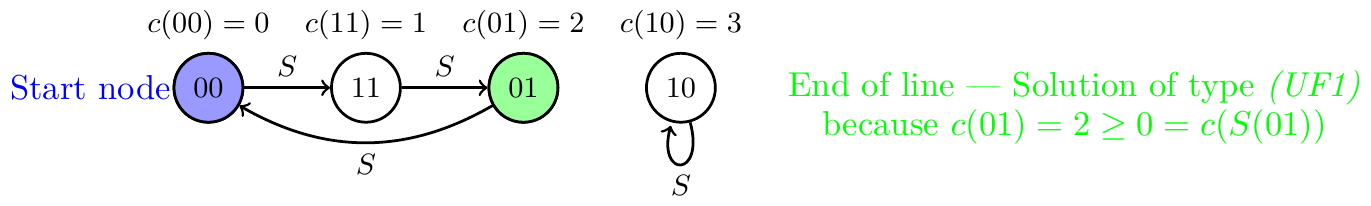}

	\caption{\UniqueForwardEOPL instances that form a valid line without violations.}
	\label{NoVios}
\end{figure}
	
\begin{figure}
Candidate solution space: \f{\dimension=3}
\centering
\includegraphics{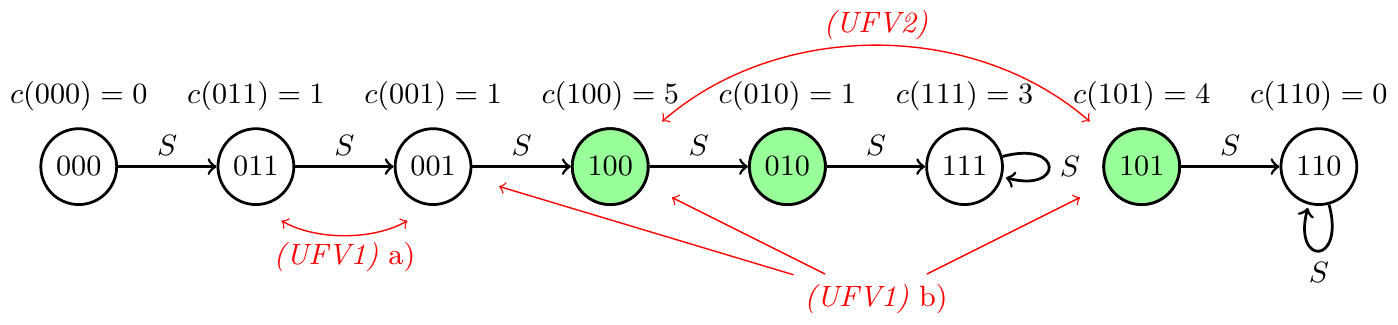}
\caption{A \UniqueForwardEOPL instance line with all types of violations.}
\label{Vios}
\end{figure}

\begin{definition}(\cite[Definition~7]{UEOPL2020})
Let $\searchProblemA$ be a search problem, and $\Instances{\searchProblemA}\subseteq \{0, 1\}^*$ be the set of all instances for \searchProblemA. 
For an instance \f{\instance \in \Instances{\searchProblemA}}, let \f{\Solutions{\searchProblem}(\instance)} be the set of candidate solutions of $\instance$.
A search problem \searchProblemA can be reduced by a \emph{promise preserving Karp reduction in polynomial time} to a search problem \searchProblemB if there exist two polynomial-time functions 
\f{f: \Instances{\searchProblemA} \rightarrow \Instances{\searchProblemB}} and
\f{g: \Instances{\searchProblemA} \times \Solutions{\searchProblemB}(f ( \instanceA )) \rightarrow \Solutions{\searchProblemA}(\instanceA )}
such that 
if \solutionB is a violation of \f{f(\instance)}, then \f{g(\instanceA, \solutionB)} is a violation of \instanceA and
if \solutionB is a valid solution of \f{f(\instance)}, then \f{g(\instanceA, \solutionB)} is a valid solution or a violation of \instanceA.
We are given an instance of \searchProblemA from which we construct with $f$ an instance of problem \searchProblemB. If we then solve \searchProblemB, we can re-translate the solution from \searchProblemB to a solution of \searchProblemA with $g$.
Promise preserving reductions are transitive.
\end{definition}

\begin{definition} (\cite{UEOPL2020})
The search problem complexity class \UEOPL contains all problems that can be reduced in polynomial time to \UniqueForwardEOPL. 
Thus, the complexity class \UEOPL captures all total search problems where the space of candidate solutions has the structure of a unique line with increasing cost. The relationship of \UEOPL to other classes is shown in Figure \ref{relationship}.
\end{definition}

\PUEOPL is the promise version of the search problem class \UEOPL. When containment of a search problem \searchProblem in \UEOPL is shown via a promise preserving reduction, the promise version of \searchProblem is contained in \PUEOPL.

\begin{figure}[ht!]
\centering
\includegraphics{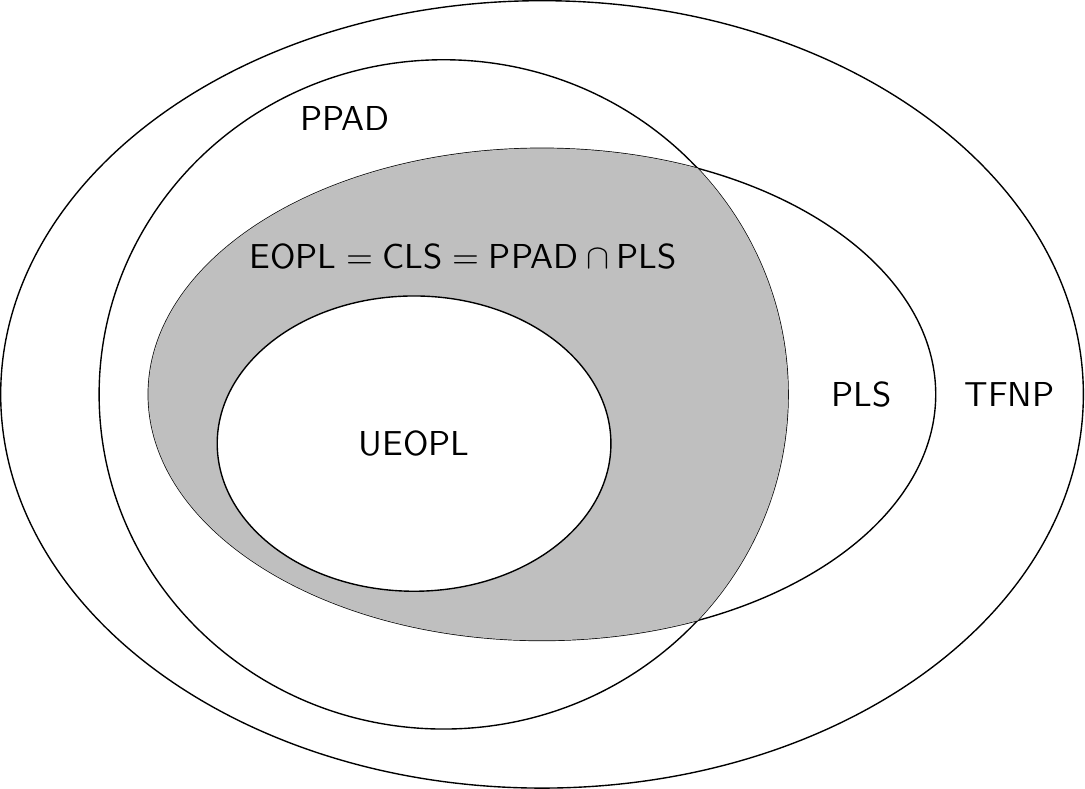}
\caption{Relation of \UEOPL to other search problem complexity classes according to \cite{EOPL=CLS}.}
\label{relationship}
\end{figure}

\section{Unique Sink Orientations of Grids}

\begin{definition}
(\cite[p.~206]{GridUSO})
Let \f{\length, \dimension \in \Naturals^+}. 
Let \f{\Menge=\{1, \dots, \length \}} be an ordered set of integers called \emph{directions} and \f{\Partitions=(\partition_1, \dots, \partition_\dimension)} be a partition of \Menge with \f{\partition_i} being ordered and \f{\abs{\partition_i}\geq 2} for all \emph{dimensions} \f{i=1, \dots, \dimension}. 
The \emph{\dimension-dimensional grid} \Grid is the undirected graph derived from \f{\Grid=(\Menge, \Partitions)} with a set of vertices
\f{\Vertices := \{ \vertex \subseteq \Menge \mid i=1, \dots, \dimension, \abs{\vertex \cap \partition_i} = 1 \}} and a set of edges
\f{\Edges := \{ \{\pointA, \pointB\} \mid \pointA, \pointB \in \Vertices, |\pointA \xor \pointB| = 2 \}},
where $\xor$ denotes the \emph{symmetric difference}.

\end{definition}

\begin{definition}
(\cite[p.~211]{GridUSO})
The \emph{outmap function} \f{\outmap: \Vertices \rightarrow \Potenzmenge{\Menge}} defines an \emph{orientation} of the edges of a grid $\Grid$.
For each point \f{\pointA \in \Vertices}, the set  $\sigma(p)$ 
contains all directions to which \pointA has outgoing edges.
The edges of \pointA for all other directions are incoming.
In particular, we have $\sigma(p) \cap p = \emptyset$.
An outmap \outmap is called \emph{unique sink orientation} of \Grid if all nonempty induced subgrids of \Grid have a unique sink.
\end{definition}

\begin{definition}
(\cite[Definition~2.13]{GridUSO})
The \emph{refined index} \f{\refinedIndex_\outmap} of an outmap \f{\outmap} with \f{\refinedIndex_\outmap\colon \Vertices \rightarrow \{0, \dots, \abs{\partition_1}-1\} \times \dots \times \{0, \dots, \abs{\partition_\dimension}-1\}} is defined as: \f{\refinedIndex_\outmap(\pointA) := (\abs{\outmap(\pointA) \cap \partition_1}, \dots, \abs{\outmap(\pointA) \cap \partition_\dimension})}.
The refined index assigns to each point a \f{\dimension}-tuple containing at index \f{i} the number of outgoing edges in dimension \f{i}.
\end{definition}

\begin{theorem}
(\cite[Theorem~2.14]{GridUSO}) If \outmap is a unique sink orientation, then \f{\refinedIndex_\outmap} is a bijection.
\end{theorem}


\begin{figure}[ht!]
\centering
\includegraphics{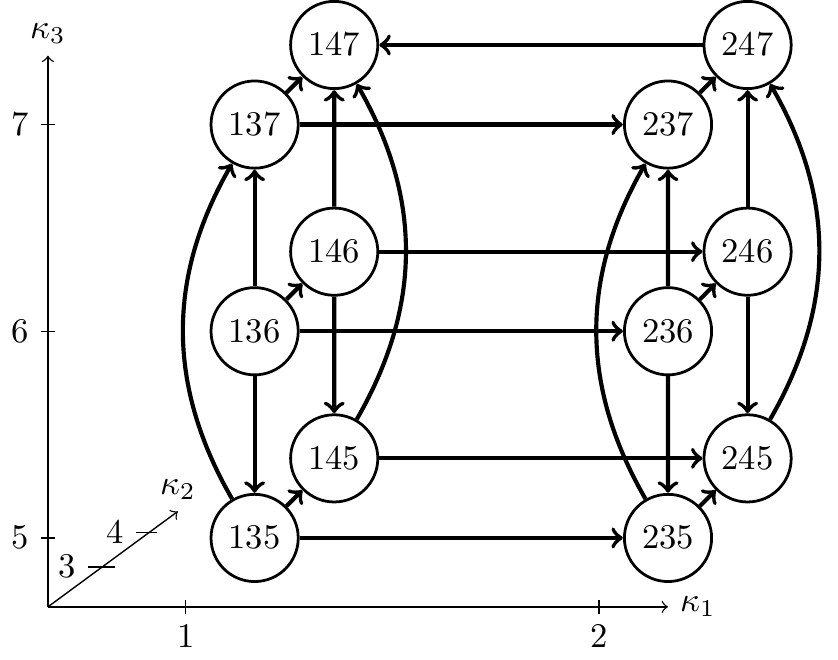}
\caption{Example grid \Grid with \f{\Menge = \{ 1, \dots, 7\}}, \f{\partition_1= \{1, 2\}}, \f{\partition_2 = \{3, 4\}} and \f{\partition_3 = \{5, 6, 7\}}. The orientation of the edges is a unique sink orientation. The unique sink is the point \f{(147)}. The outmap of point \f{(246)} is \f{\outmap(246) = \{ 5, 7 \}} and its refined index is \f{\refinedIndex_\outmap(246) = (0, 0, 2)}.}
\label{grid}
\end{figure}


\begin{definition} 
(\cite[Definition~6.1.23]{MA})
The search problem \GridUSO is defined as follows:
Given a \dimension-dimensional grid, represented implicitly by \f{\Grid = (\Menge, \Partitions)}, and a circuit computing an outmap function \f{\outmap \colon \Vertices \rightarrow \Potenzmenge{\Menge}}, find one of the following:
\begin{itemize}
\setlength{\itemindent}{1cm}

\item[\GridUSOSolTypeEndOfLine] A point \f{\pointA \in \Vertices} with \f{\outmap(\pointA) = \emptyset}.
The point \pointA is a sink.

\item[\GridUSOVioSelfLoop] A point \f{\pointA \in \Vertices} with \f{\pointA \cap \outmap(\pointA) \neq \emptyset}.
The point \pointA has a directed edge to itself, thus it is a certificate of $\sigma$ not being a valid unique sink orientation.

\item[\GridUSOVioRefinedIndex]  An induced subgrid $\subGrid = (\Menge', \Partitions')$ with \f{\outmap'(\pointA) := \outmap(\pointA) \cap \Menge'} and two points \f{\pointA, \pointB \in \Vertices'} with \f{\pointA \neq \pointB} and \f{\refinedIndex_{\outmap'}(\pointA) = \refinedIndex_{\outmap'}(\pointB)}.
The points \pointA, \pointB and the subgrid \subGrid are a polynomial time verifiable certificate 
that \f{\refinedIndex_{\outmap'}} is not a bijection and thus, \outmap not a unique sink orientation.

\end{itemize}

\end{definition}
	
\GridUSO is a total search problem.
Under the promise that the outmap \outmap is a unique sink orientation, the unique sink will be found and returned as solution \GridUSOSolTypeEndOfLine. 
If \outmap is \emph{not} a unique sink orientation, there exists at least one of the violations \GridUSOVioSelfLoop or \GridUSOVioRefinedIndex. 
An example instance can be seen in Figure \ref{grid}.
Unique sink orientations on grids were introduced by Gärtner et al. \cite{GridUSO} as a combinatorial abstraction of linear programming over products of simplices and the generalized linear complementarity problems over P-matrices.
There is no polynomial time algorithm known to solve \GridUSO.


\section{\GridUSO is in \UEOPL}

\begin{theorem}
\GridUSO can be reduced via a promise preserving reduction to \UniqueForwardEOPL.
\end{theorem}
\begin{proof}

Given an instance \f{\instanceA=(\Grid, \outmap)} of \GridUSO, we construct one instance of  \UniqueForwardEOPL \f{\instanceB = (\Successor, \cost)} such that solutions and violations can be mapped accordingly.


\subparagraph{Idea of the reduction.}

We construct a line following algorithm that finds for any given grid its unique sink or a violation, such that each step can be calculated in polynomial time.
The vertices of \instanceB are encodings of the states of this line following algorithm. The successor function \Successor calculates the next state.
To do so, it is allowed to call the outmap function \outmap from the \GridUSO instance.


\begin{lemma}
\label{oneOfThemIsSink}
(\cite[Lemma~6.3.5]{MA})
Given an outmap \outmap on a grid \Grid, two disjoint induced subgrids \GridA and \GridB of \Grid whose union is again a valid subgrid, and their respective unique sinks \blueSink and \yellowSink, then \emph{(a)} the unique sink of the grid \f{(\GridA \cup \GridB)} is either \blueSink or \yellowSink or \emph{(b)} \outmap is not a unique sink orientation and we found a violation of \GridUSO.
\end{lemma}


\begin{figure}[ht!]
\centering
\includegraphics{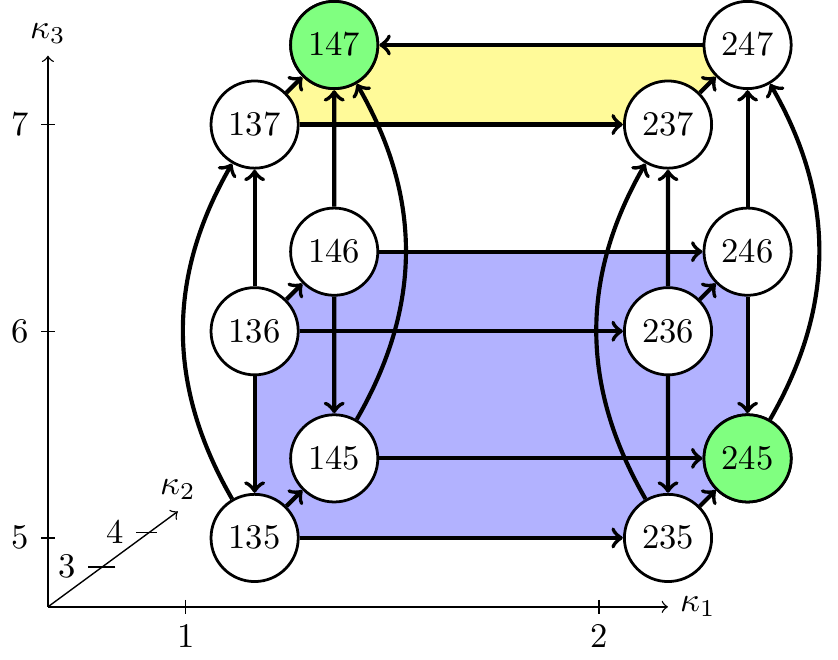}
\caption{Step of the line-following algorithm: either \f{(245)} or \f{(147)} is the sink of the whole grid.}
\label{subgridsAndSinks}
\end{figure}


The line following algorithm starts at the bottom left point of the grid.
It iterates over all directions, in each step looking at the subgrid that is formed by the union of the previous directions (e.g., the subgrid colored in blue in Figure~\ref{subgridsAndSinks}) and its sink \blueSink.
By adding the next lexicographic direction,
we add another subgrid (e.g., the yellow subgrid in Figure~\ref{subgridsAndSinks}) for which we can find the sink \yellowSink recursively.
The sink of the grid which is the union of the blue and the yellow subgrids is either \blueSink or \yellowSink. 
None of the other nodes in the subgrids is a sink of the respective subgrid, and thus not a sink of the combined grid.
If neither \blueSink nor \yellowSink is a sink, then it can be proven that there is a refined index violation of type \GridUSOVioRefinedIndex.

\newcommand{\FindSink}{\texttt{find\_sink}}

\begin{algorithm}[H]

	\label{LineFollowingAlgo}
	\caption{\FindSink \f{(\outmap, \Menge, \partition_1, \dots, \partition_\dimension)}}		
	\blueSink := \f{((\partition_{1})_1, \dots, (\partition_\dimension)_1)}\;
	\For{\f{i \in \Menge}}{
	\label{step1}
		\If{\f{i \notin \outmap(\blueSink)} \tcp*[f]{\blueSink is also sink in $i$'th subgrid}}{ 
			\label{check1}
			continue with next \f{i}\;
		}
		Let \f{j} be the index such that \f{i \in \partition_j} \;
		\f{\Menge'} := \f{(\Menge \setminus \partition_j) \cup \{ i \}} \;
		\yellowSink := \FindSink(\f{\outmap, \Menge', (\partition_1 \cap \Menge'), \dots, (\partition_\dimension \cap \Menge') }) \tcp*[f]{Search recursively for sink in yellow subgrid} \;  
		\For{\f{k \in \partition_j \wedge k \leq i}}{
			\label{step2}
			\If{\f{k \in \outmap(\yellowSink)} \tcp*[f]{If \yellowSink is not sink in $i$'th subgrid}}{
				return Violation\;
			}
		}
		\blueSink := \yellowSink\;
	}
	return \blueSink\;
\end{algorithm}

\subparagraph{Construction of the vertices.}

Each node of \instanceB is the bit-encoding of an $(\length+1)$-tuple of points of the grid. Let \f{\Vertices' := ( \Vertices \cup \{ \bot \})^{\length + 1}}.
Its contents encode the state of Algorithm \ref{LineFollowingAlgo}.
The points stored in the vertices are the unique sinks of the subgrids in which Algorithm \ref{LineFollowingAlgo} searches recursively. 
The position corresponds to the directions through which the algorithm iterates. Thus, we can identify for each tuple which step of the algorithm it represents.
Let the start node be the $(\length+1)$-tuple consisting of the bottom left point of the grid and \length many $\bot$'s: \f{(((\partition_1)_1, \dots, (\partition_\dimension)_1), \bot, \dots, \bot)}. 
We define the function \isVertex which checks in polynomial time, whether any given $(\length+1)$-tuple is a valid step of Algorithm \ref{LineFollowingAlgo}, the representation of a \GridUSO violation or not a proper encoding at all.

\subparagraph{Construction of the successor function \Successor.}

The successor function \Successor checks for the given node whether it encodes a valid step of Algorithm \ref{LineFollowingAlgo} and if so, calculates the next step in polynomial time.
Each step, including the steps of the recursive calls, is a separate node on the resulting unique line.
The line corresponds to traversing the tree of the call hierarchy.
Let \f{\Successor \colon \Vertices' \rightarrow \Vertices'}.
Given a vertex \f{\nodeA = (\pointA_1,  \dots, \pointA_\length, \pointA_{\length +1})}, let \f{\Successor(\nodeA)} be:
\begin{enumerate}
	\item If \f{\isVertex(\nodeA)} says \nodeA is \emph{not} a valid encoding, then set \f{\Successor(\nodeA) := \nodeA}.
	\item If \f{\isVertex(\nodeA)} says \nodeA is a valid encoding and not a violation:
	\begin{enumerate}
		\item If the node has the form \f{\nodeA= (\bot, \dots, \bot, \pointA_{\length +1})}
		it encodes the end of the for-loop in line \ref{step1} of Algorithm \ref{LineFollowingAlgo}. Thus, the algorithm returns the sink \f{\pointA_{\length +1}}. Thus, set \f{\Successor(\nodeA):=\nodeA} to indicate the end of the line.
		
		\item If the node has the form \f{\nodeA= (\bot, \dots, \bot, \pointA_{i},\pointA_{i+1},\pointA_{i+2}, \dots, \pointA_{\length +1})}, then \f{\pointA_{i}} is the sink \blueSink of the \f{i}'th iteration of the for-loop in line \ref{step1}.

		\begin{enumerate}
			\item  If the check in line \ref{check1} is true, we know that \blueSink is also the sink of the next bigger subgrid. Thus, set \f{\Successor(\nodeA):=(\bot, \dots, \bot,\bot, \pointA_{i},\pointA_{i+2}, \dots, \pointA_{\length +1})}.

			\item If the check in line \ref{check1} is false,
			then we know that \blueSink is \emph{not} the sink of the next bigger subgrid. By Lemma \ref{oneOfThemIsSink}, we must search recursively for the sink of the yellow subgrid.
			But we want to remember \blueSink, so that we can identify a violation if one exists.
			Thus, set \f{\Successor(\nodeA):=(s, \bot, \dots, \bot, \pointA_{i},\pointA_{i+1}, \dots, \pointA_{\length +1})}, where \f{s} is the start point of the subgrid in the recursive call.
		\end{enumerate}	
	\end{enumerate}
\end{enumerate}

\subparagraph{Construction of the cost function \cost.}
Let $\gridwidth = \length +2$ and \f{\help \colon \Vertices' \times \{ 1, ..., \length \} \times \{ 1, \dots, \dimension \}\rightarrow \{0,\dots,\gridwidth -1 \}} a help function with
\begin{equation}
\help(\node, \ii, \ij) := \begin{cases}
0 & \text{if }\point_\ii= \bot, \\
\gridwidth - 1 & \text{if } \outmap(\point_\ii) \cap \{ 1, \dots, \ii \} = \emptyset, \text{i.e., step \emph{(2.b.i)} holds,}\\
(\point_\ii)_\ij  & \text{otherwise.}\end{cases}		
\end{equation}
Because the grid works with ordered sets, lexicographically bigger points have a higher value in \help. 
Also, if two points are sink of their corresponding subgrid, they have the same value in \help.
The help values of each point and every dimension are scaled and summed up.
The cost of a node then is the sum of these scaled values,
where each summand is scaled again to enforce that the value for a point at position $i$ which is not $\bot$ is always higher than the sum of all values of points with indices smaller than $i$.
Thus, we always give steps later on in the algorithm higher cost.
\begin{equation}
\cost(\vertex) :=  \left( \sum_{i = 1}^{\length} \left( \gridwidth^{i-1} \cdot \sum_{j = 1}^{\dimension} \gridwidth^j \help(\node, i, j) \right) \right) + \begin{cases}
		0 & \text{if } \point_{\length + 1} = \bot, \\
		\gridwidth^{\length \dimension + 1} & \text{if } \point_{\length + 1} \neq \bot.
\end{cases}
\end{equation}


\subparagraph{Correctness}
The successor function and the cost function can be constructed in polynomial time.
It can be proven that every solution of type \GridUSOSolTypeEndOfLine is only mapped to solutions of type \UFEOPLSolTypeEndOfLine. 
Every violation of the created \UniqueForwardEOPL instance can be mapped back to a violation of the \GridUSO instance.
Therefore this reduction is promise preserving.
It follows that \GridUSO is in \UEOPL and the promise version of \GridUSO is in \PUEOPL.
The full proof can be found in \cite[Proof of Theorem~6.3.1]{MA}.
\end{proof}

\bibliography{literature}

\end{document}

%% file: generalSettings.tex
\usepackage{xspace}
\usepackage{amsthm}
\usepackage[noend, ruled, boxed, linesnumbered]{algorithm2e}

\newcommand{\f}[1]{\relax\ifmmode#1\else{$#1$}\fi}
\newcommand{\abs}[1]{ \lvert#1\rvert }
\newcommand{\xor}{\oplus}

\newcommand\Naturals{\mathbb{N}\xspace}

\newcommand{\Class}[1]{\textsf{#1}}
\newcommand{\UEOPL}{\Class{UEOPL}\xspace}
\newcommand{\PUEOPL}{\Class{PromiseUEOPL}\xspace}
\newcommand{\TFNP}{\Class{TFNP}\xspace}

\newcommand{\NP}{\Class{NP}\xspace}
\newcommand{\myProblem}[1]{\textsc{#1}}

\newcommand{\GridUSO}{\myProblem{Grid-USO}\xspace}

\newcommand{\UniqueForwardEOPL}{\myProblem{Unique Forward EOPL}\xspace}

\newcommand{\UFEOPLSolTypeEndOfLine}{\emph{(UF1)}\xspace}
\newcommand{\UFEOPLVioTypeNotALine}{\emph{(UFV1)}\xspace}
\newcommand{\UFEOPLVioBreakInLine}{\emph{(UFV2)}\xspace}

\newcommand{\GridUSOSolTypeEndOfLine}{\emph{(GU1)}\xspace}
\newcommand{\GridUSOVioSelfLoop}{\emph{(GUV1)}\xspace}
\newcommand{\GridUSOVioRefinedIndex}{\emph{(GUV2)}\xspace}

\newcommand{\Successor}{\f{S}\xspace}

\newcommand{\cost}{\f{c}\xspace}
\newcommand{\dimension}{\f{d}\xspace}

\newcommand{\length}{\f{n}\xspace}
\newcommand{\Grid}{\f{\Gamma}\xspace}
\newcommand{\gridwidth}{\f{\omega}\xspace}
\newcommand\Graph{\f{G}\xspace}

\newcommand\Vertices{\f{V}\xspace}
\newcommand\node{\f{v}\xspace}
\newcommand\vertex{\f{v}\xspace}
\newcommand\Edges{\f{E}\xspace}

\newcommand{\nodeA}{\f{v}\xspace}
\newcommand{\nodeB}{\f{w}\xspace}

\newcommand{\Menge}{\f{\mathcal{M}}\xspace}
\newcommand{\Partitions}{\f{K}\xspace}
\newcommand{\partition}{\f{\kappa}\xspace}

\newcommand{\subGrid}{\f{\Grid'}\xspace}

\newcommand{\outmap}{\f{\sigma}\xspace}

\newcommand{\Potenzmenge}[1]{\f{Powerset(#1)}\xspace}
\newcommand{\refinedIndex}{\f{r}\xspace}

\newcommand{\point}{\f{p}\xspace}
\newcommand{\pointA}{\f{p}\xspace}
\newcommand{\pointB}{\f{q}\xspace}
\newcommand{\isVertex}{\texttt{isVertex}\xspace}

\newcommand{\blueSink}{\f{x}\xspace}
\newcommand{\yellowSink}{\f{y}\xspace}

\newcommand{\GridA}{\f{\Grid'}\xspace}
\newcommand{\GridB}{\f{\Grid''}\xspace}

\newcommand{\ii}{\f{i}\xspace}
\renewcommand{\ij}{\f{j}\xspace}

\newcommand{\Instances}[1]{\f{\mathcal{I}^{#1}}\xspace}
\newcommand{\instance}{\f{I}\xspace}
\newcommand{\Solutions}[1]{\f{\mathcal{S}^{#1}}\xspace}

\newcommand{\searchProblem}{\f{R}\xspace}

\newcommand{\searchProblemA}{\f{\searchProblem}\xspace}
\newcommand{\searchProblemB}{\f{\searchProblem'}\xspace}

\newcommand{\solutionB}{\f{\solution'}\xspace}

\newcommand{\instanceA}{\f{\instance}\xspace}
\newcommand{\instanceB}{\f{\instance'}\xspace}

\newcommand{\help}{\f{h}\xspace}